\documentclass[11pt]{article}

\usepackage{amsmath, amssymb, amsthm}
\usepackage{mathtools}
\usepackage{bm}
\usepackage{graphicx}
\usepackage{hyperref}
\usepackage{enumitem}
\usepackage{thmtools}
\usepackage{thm-restate}
\usepackage{booktabs}
\usepackage{tikz}
\usepackage{float}
\usepackage{multirow}

\theoremstyle{plain}
\newtheorem{theorem}{Theorem}[section]
\newtheorem{lemma}[theorem]{Lemma}
\newtheorem{proposition}[theorem]{Proposition}
\newtheorem{corollary}[theorem]{Corollary}

\theoremstyle{definition}
\newtheorem{definition}[theorem]{Definition}
\newtheorem{example}{Example}[section]

\theoremstyle{remark}

\usetikzlibrary{arrows.meta, positioning}

\hypersetup{
    colorlinks=true,
    linkcolor=blue,
    citecolor=blue,
    urlcolor=blue
}

\title{A Continuous Multi-Component Measure of Directed Acyclicity (DAG-ness)}
\author{\normalsize Erik Csikos$^{1,2}$\\
\small $^1$Binghamton University, Binghamton, United States\\
\small $^2$Moravian University, Bethlehem, United States\\}

\date{}

\begin{document}

\maketitle
\begin{abstract}
Directed acyclic graphs (DAGs) are fundamental to the study of causal structures, hierarchical systems, and information flow. While directedness and acyclicity are defined as binary properties, real--world networks often exhibit continuous degrees of "DAG-ness" due to structural noise, back-edges, or localized feedback loops. Our previous attempt to quantify DAG-ness as a continuous measure suffered from topological redundancy, where overlapping cyclic penalties artificially deflated scores for networks with minor feedback. In this paper, we resolve these limitations by introducing a strictly orthogonal, 4-dimensional continuous DAG-ness framework. By independently measuring the volume of feedback $A(G)$, the alignment of flow $F(G)$, the macroscopic locality of feedback $M(G)$, and dynamical pathway complexity $S(G)$, the proposed measure eliminates collinearity and the "Dilution Trap." Empirical evaluation on synthetic diagnostic graphs demonstrates enhanced mathematical stability, while deterministic application to classical number--theoretic systems (the Kaprekar and Collatz graphs) confirms the framework's ability to rigorously isolate topological flow from dynamical entrapment. The resulting composite score $D(G)$ provides a highly scalable, interpretable, and mathematically sound metric for structural network analysis.
\end{abstract}

\section{Introduction}

The directed acyclic graph (DAG) is a ubiquitous structural motif in network science, serving as the mathematical foundation for causal inference, software dependency mapping, biological food webs, and organizational hierarchies. In pure graph theory, a DAG is defined by a strict binary condition: the complete absence of directed cycles. However, when analyzing complex, real--world systems, this binary classification is often highly brittle. A massive, perfectly ordered hierarchical network containing millions of nodes can be instantly classified as a non-DAG by the introduction of a single, localized feedback loop or noisy symmetric edge. 

Consequently, there is a significant need for a continuous measure of ``DAG-ness'', that is, a metric that can quantify exactly how closely a given directed network resembles a pure topological DAG on a scale from $0$ to $1$. 

Previous work by the author introduced a foundational 5-component framework designed to capture various facets of structural hierarchy and acyclicity \cite{csikos2026}. This model served as a vital conceptual precursor, successfully demonstrating that DAG-ness could be evaluated continuously. However, because it relied heavily on strongly connected component (SCC) surrogates and trace--exponential matrices \cite{eades, zheng}, it exhibited unintended mathematical collinearity. A single localized cycle could trigger simultaneous penalties across node-level, macro--level, and hierarchical dimensions. Furthermore, the reliance on continuous algebraic walks exposed the initial framework to severe score degradation when evaluating simple undirected trees.

In this paper, we present a refined, strictly orthogonal 4-dimensional framework that builds upon this conceptual foundation. By decoupling the static topological structure of the graph from its dynamic spectral flow, the revised measure eliminates component redundancy while preserving the empirical validity of the original approach.

The remainder of the paper develops this refined perspective in a structured progression, beginning with the mathematical foundations of directed acyclicity. The structure of the paper is as follows. Section 2
 establishes the structural, algebraic, and spectral foundations of directed acyclicity. Section 3 introduces the refined four‑component DAG‑ness framework and formalizes each component. Section 4
 provides a theoretical justification of the framework, demonstrating how it eliminates redundancy and resolves the limitations of earlier formulations. Section 5 analyzes the behavior of the measure on number‑theoretic dynamical systems, and Section 6 presents empirical validation on synthetic diagnostic graph families.

\section{Mathematical Foundations of Directed Acyclicity}

Directed acyclic graphs (DAGs) arise naturally in settings where information,
influence, or computation flows in a consistent direction. Their defining
properties, directionality and the absence of directed cycles, support
hierarchical organization, causal inference, and efficient algorithmic processing
\cite{diestel, bangjensen-gutin}. In this section, we develop the mathematical
foundations of directed acyclicity, establishing the structural concepts that
motivate a continuous measure of DAG-like behavior.

\subsection{Topological and Combinatorial Foundations}

\begin{definition}[Directed Graph]
A \emph{directed graph} (or \emph{digraph}) is a pair $G = (V,E)$ where $V$ is a
finite set of vertices and $E \subseteq V \times V$ is a set of directed edges.
\end{definition}

\begin{definition}[Walks and Cycles]
A \emph{directed walk} of length $k$ is a sequence $(v_0, v_1, \dots, v_k)$ such
that $(v_{i-1}, v_i) \in E$ for all $i$. A walk is a \emph{cycle} if $v_0 = v_k$
and $k \ge 1$.
\end{definition}

\begin{definition}[Reachability]
For $u,v \in V$, we write $u \leadsto v$ if there exists a directed walk from $u$
to $v$. The relation $\leadsto$ is called the \emph{reachability relation}.
\end{definition}

\begin{proposition}
The reachability relation $\leadsto$ is a preorder on $V$: it is reflexive and
transitive.
\end{proposition}

\begin{proof}
Reflexivity follows from the trivial walk $(v)$ of length $0$. Transitivity
follows by concatenating walks. \end{proof}

The failure of antisymmetry in $\leadsto$ corresponds precisely to the presence
of directed cycles.

\begin{definition}[Strongly Connected Components]
A subset $S \subseteq V$ is a \emph{strongly connected component} (SCC) if every
pair of vertices in $S$ is mutually reachable and $S$ is maximal with respect to
this property.
\end{definition}

The SCCs of a digraph partition the vertex set, and their structure is fundamental
to understanding directed cycles \cite{bangjensen-gutin}. The condensation graph
$C(G)$, introduced in classical digraph theory \cite{bangjensen-gutin, tarjan1972},
encodes the global feedback structure of $G$.

\begin{definition}[Condensation DAG]
Let $G = (V,E)$ be a digraph with SCCs $S_1, \dots, S_k$. The \emph{condensation
graph} $C(G)$ is the directed graph whose vertices correspond to the SCCs and
with an edge $S_i \to S_j$ if there exists $(u,v) \in E$ with $u \in S_i$ and
$v \in S_j$.
\end{definition}

\begin{theorem}
The condensation graph $C(G)$ is always a DAG \cite{bangjensen-gutin}.
\end{theorem}

\begin{proof}
Suppose $C(G)$ contains a directed cycle $S_{i_1} \to S_{i_2} \to \cdots \to
S_{i_m} \to S_{i_1}$. Then for any $u \in S_{i_1}$ and $v \in S_{i_j}$, the
edges in the cycle imply $u \leadsto v$ and $v \leadsto u$, so all vertices in
the cycle belong to a single SCC, contradicting maximality. \end{proof}

Thus, a directed graph is acyclic if and only if all SCCs are trivial.

\begin{definition}[Partial Order]
A relation $\preceq$ on $V$ is a \emph{partial order} if it is reflexive,
antisymmetric, and transitive.
\end{definition}

\begin{theorem}[DAGs and Partial Orders]
A directed graph $G$ is acyclic if and only if the reachability relation
$\leadsto$ is a partial order.
\end{theorem}

\begin{proof}
If $G$ is acyclic, then $\leadsto$ is antisymmetric: if $u \leadsto v$ and
$v \leadsto u$, then $u$ and $v$ lie on a directed cycle, contradicting
acyclicity. Conversely, if $\leadsto$ is antisymmetric, then no directed cycle
can exist, since a cycle would imply mutual reachability. \end{proof}

This correspondence allows DAGs to be viewed as combinatorial realizations of
finite partially ordered sets.

A permutation $\pi$ of $V$ is a \emph{topological ordering} if every edge
$(u,v) \in E$ satisfies $\pi(u) < \pi(v)$. The existence of such an ordering is
equivalent to acyclicity, a classical result due to Kahn \cite{kahn1962} and
Tarjan \cite{tarjan1972}.

\begin{theorem}
A directed graph $G$ is acyclic if and only if it admits a topological ordering
\cite{kahn1962, tarjan1972}.
\end{theorem}

Violations of a topological order provide a natural way to quantify the degree to
which a graph deviates from acyclicity. This idea underlies several hierarchy
measures \cite{mones, moutsinas-evans}.

\subsection{Algebraic and Spectral Characterizations}

Let $A$ be the adjacency matrix of $G$. The powers of $A$ encode walk counts, a
standard result in algebraic graph theory \cite{diestel, bangjensen-gutin,
horn-johnson}.

\begin{proposition}
For all $k \ge 1$, the entry $(A^k)_{ij}$ equals the number of directed walks of
length $k$ from vertex $i$ to vertex $j$.
\end{proposition}

\begin{proof}
The result follows by induction on $k$ using matrix multiplication
\cite{horn-johnson}. \end{proof}

The trace of $A^k$ counts closed walks of length $k$, a classical fact in matrix
analysis \cite{horn-johnson, meyer2000}.

\begin{theorem}[Cycle Detection via Matrix Powers]
A directed graph $G$ contains a directed cycle if and only if
$\operatorname{trace}(A^k) > 0$ for some $1 \le k \le n$
\cite{horn-johnson, meyer2000}.
\end{theorem}

\begin{definition}[Nilpotent Matrix]
A matrix $A$ is \emph{nilpotent} if $A^m = 0$ for some integer $m \ge 1$.
\end{definition}

Nilpotency of the adjacency matrix is equivalent to acyclicity, a well-known
characterization in algebraic graph theory \cite{horn-johnson, meyer2000}.

\begin{theorem}[Acyclicity and Nilpotency]
A directed graph $G$ is acyclic if and only if its adjacency matrix $A$ is
nilpotent.
\end{theorem}

\begin{proof}
If $G$ is acyclic, then the longest directed walk has length at most $n-1$, so
$A^n = 0$. Conversely, if $A$ is nilpotent, then $A^k = 0$ for some $k$, so
$\operatorname{trace}(A^k) = 0$, implying by the previous theorem that no cycle
exists. \end{proof}

The spectral radius $\rho(A)$ provides a global indicator of recurrence, with
connections to directed cycles established in spectral graph theory
\cite{chung, estrada-recurrence}.

\begin{definition}[Spectral Radius]
For a matrix $A$, the \emph{spectral radius} is
\[
\rho(A) = \max\{ |\lambda| : \lambda \text{ is an eigenvalue of } A \}.
\]
\end{definition}

\begin{lemma}
If $G$ is acyclic, then $\rho(A) = 0$.
\end{lemma}

\begin{proof}
If $G$ is acyclic, then $A$ is nilpotent, so all eigenvalues are zero. \end{proof}

\begin{proposition}
If $G$ contains a directed cycle, then $\rho(A) \ge 1$.
\end{proposition}

\begin{proof}
A directed cycle of length $k$ contains a $k \times k$ permutation submatrix with
spectral radius $1$. Since the spectral radius of a matrix is at least that of any
principal submatrix, the result follows \cite{horn-johnson}. \end{proof}

\begin{lemma}
If $G$ is acyclic, then there exists a permutation matrix $P$ such that
$P A P^{-1}$ is strictly upper triangular.
\end{lemma}

\begin{proof}
Let $\pi$ be a topological ordering of $G$. Reordering the vertices according to
$\pi$ produces an adjacency matrix with all edges pointing forward, hence strictly
upper triangular. \end{proof}

\begin{corollary}
A directed graph is acyclic if and only if its adjacency matrix is similar to a
strictly upper triangular matrix.
\end{corollary}

Violations of this triangular structure quantify the extent to which a graph
deviates from hierarchical flow, motivating hierarchy-based measures
\cite{mones, moutsinas-evans}.

The results in this section establish several equivalent characterizations of
acyclicity:

\begin{itemize}
    \item combinatorial: absence of directed cycles,
    \item structural: all SCCs are trivial,
    \item order--theoretic: reachability is a partial order,
    \item algebraic: adjacency matrix is nilpotent,
    \item spectral: spectral radius is zero,
    \item matrix-theoretic: adjacency matrix is similar to a strictly upper triangular matrix.
\end{itemize}

\subsection{Illustrative Examples}

To complement the structural and algebraic characterizations above, we present
several small examples that illustrate the key concepts.

\begin{example}[A Simple DAG]
Consider the directed graph with vertex set $V = \{1,2,3,4\}$ and edges
$E = \{(1,2), (1,3), (2,4), (3,4)\}$. This graph is acyclic and admits the
topological ordering $(1,2,3,4)$.

\begin{center}
\begin{tikzpicture}[>=Stealth, node distance=1.6cm]
\node[circle, draw] (1) {1};
\node[circle, draw, right of=1] (2) {2};
\node[circle, draw, below of=2] (4) {4};
\node[circle, draw, below of=1] (3) {3};

\draw[->] (1) -- (2);
\draw[->] (1) -- (3);
\draw[->] (2) -- (4);
\draw[->] (3) -- (4);
\end{tikzpicture}
\end{center}

The adjacency matrix is strictly upper triangular after reordering, and $A^k = 0$
for all $k \ge 3$, confirming nilpotency.
\end{example}

\begin{example}[A Graph with One Nontrivial SCC]
Let $G$ have vertices $\{1,2,3,4\}$ and edges
$(1,2), (2,3), (3,1), (3,4)$. The vertices $\{1,2,3\}$ form a nontrivial SCC,
while $\{4\}$ is trivial.

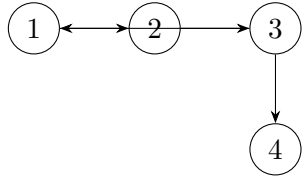
\begin{center}
\begin{tikzpicture}[>=Stealth, node distance=1.6cm]
\node[circle, draw] (1) {1};
\node[circle, draw, right of=1] (2) {2};
\node[circle, draw, right of=2] (3) {3};
\node[circle, draw, below of=3] (4) {4};

\draw[->] (1) -- (2);
\draw[->] (2) -- (3);
\draw[->] (3) -- (1);
\draw[->] (3) -- (4);
\end{tikzpicture}
\end{center}

The condensation graph consists of two vertices, one representing the SCC
$\{1,2,3\}$ and one representing $\{4\}$, with a single edge between them. Since
the condensation graph is acyclic, all cycles in $G$ are contained within the
nontrivial SCC.
\end{example}

\begin{example}[Cycle Detection via Matrix Powers]
For the graph in the previous example, the adjacency matrix $A$ satisfies
$\operatorname{trace}(A^3) = 3$, corresponding to the three closed walks of
length $3$ in the cycle $1 \to 2 \to 3 \to 1$. This illustrates the use of matrix
powers in detecting cyclic structure \cite{horn-johnson, meyer2000}.
\end{example}

\subsection{Algorithmic Considerations and the Need for a Continuous Measure}

The structural characterizations of acyclicity discussed above correspond to
efficient algorithms widely used in graph theory and computer science.

\begin{itemize}
    \item \textbf{SCC decomposition} can be computed in linear time
    $O(|V| + |E|)$ using Tarjan's algorithm \cite{tarjan1972} or Kosaraju's
    algorithm.

    \item \textbf{Topological sorting} can also be computed in linear time using
    Kahn's algorithm \cite{kahn1962} or depth-first search.

    \item \textbf{Cycle detection} via DFS is linear, while matrix-based cycle
    detection using $\operatorname{trace}(A^k)$ is conceptually useful but
    computationally expensive for large graphs.

    \item \textbf{Spectral radius estimation} can be performed efficiently using
    power iteration or Lanczos methods, which require only sparse
    matrix-vector products \cite{chung}.
\end{itemize}

These algorithmic results highlight that many structural properties of directed
graphs can be computed at scale, a fact that will be important in the development
and validation of the DAG-ness measure.

Although acyclicity is a well--defined binary property, real-world networks often
exhibit intermediate behavior that is not captured by a simple yes/no test.

\begin{itemize}
    \item A graph may contain a single small feedback loop embedded in an
    otherwise hierarchical structure.

    \item A graph may have large SCCs but still exhibit strong directional flow
    outside those components.

    \item Spectral properties may indicate global recurrence even when local
    cycles are small or sparse.

    \item Hierarchical measures may be high even when a few back-edges violate
    perfect ordering.
\end{itemize}

Binary acyclicity collapses all such cases into the same category: ``not a DAG.''
This motivates the need for a continuous measure that captures the \emph{degree}
and \emph{nature} of deviation from acyclicity.

The equivalences established in this section show that acyclicity manifests
simultaneously in combinatorial, algebraic, spectral, and order-theoretic forms.
Each perspective highlights a different structural feature of directed graphs.

In Section~3, we introduce a multi-component framework that quantifies
DAG-like behavior along several orthogonal dimensions. These components are
designed to reflect the structural, spectral, and hierarchical properties developed
in this section, providing a continuous and interpretable measure of DAG-ness.
\section{Refined DAG-ness Framework: A 4-Dimensional Measure for Directed Acyclicity}

Let $G=(V,E)$ be a finite directed graph. While strict acyclicity is a binary property, real-world networks and complex dynamical systems often exhibit continuous degrees of DAG-like behavior. To quantify this, we introduce a continuous DAG-ness framework consisting of four mathematically orthogonal components. 

To ensure the framework remains structurally sound when analyzing systems containing symmetric relationships, we formally define a filtered graph $G_{>2}$. 

\textbf{Definition 3.1 (Filtered Graph).} Let $G=(V,E)$ be a directed graph. The filtered graph $G_{>2}$ is obtained by removing every mutually symmetric edge pair. That is, if both $(u,v) \in E$ and $(v,u) \in E$, both edges are removed from $E$. 

This formal filtering preserves all cycles of length 3 or greater while strictly removing 2-cycles. By evaluating the topological components $A(G)$ and $M(G)$ exclusively on $G_{>2}$, the framework forgives simple two-way diffusion while rigorously penalizing macroscopic feedback loops.

All four components take values in the interval $[0, 1]$, with larger values indicating greater DAG-likeness, and are designed to be scalable to large networks.

\subsection{Volume of Feedback (Acyclicity)}

The first component measures the fundamental structural presence of cyclic feedback. Rather than aggregating node-level participation, which can artificially inflate due to overlapping cycles, we utilize the Minimum Feedback Arc Set (MFAS). This quantifies the minimum intervention required to break all cycles of length 3 or greater. 

Because exact MFAS computation is NP-hard, we employ the scalable linear-time heuristic proposed by Eades, Lin, and Smyth (ELS). The volume of feedback, $A(G)$, is defined as:

$$A(G) = 1 - \frac{\text{MFAS}_{ELS}(G_{>2})}{|E_{>2}|}$$

This quantity equals 1 precisely when the macroscopic structure of the graph is entirely free of structural feedback loops.

\subsection{Alignment of Flow (Directedness)}

A graph may be perfectly acyclic but still lack the progressive hierarchical flow that defines a true DAG. The second component evaluates how well the raw network adheres to a macroscopic topological gradient. 

Let $\pi: V \rightarrow \mathbb{R}$ represent a global hierarchical ranking of the vertices. To ensure computational consistency and efficiency across the framework, we define $\pi$ as the linear ordering produced by the ELS heuristic utilized in Section 3.1. The alignment of flow, $F(G)$, measures the proportion of edges in the raw graph $G$ that successfully transition from a higher rank to a lower rank:

$$F(G) = \frac{|\{(u,v) \in E : \pi(u) < \pi(v)\}|}{|E|}$$

By evaluating the raw graph $G$, this component inherently penalizes symmetric 2-cycles, as information cannot simultaneously flow ``downhill'' in both directions across a symmetric edge pair.

\subsection{Locality of Feedback (Macro-Traps)}

The third component distinguishes between networks with minor, localized stuttering and networks containing catastrophic structural black holes. A strongly connected component (SCC) is a maximal vertex set where every pair is mutually reachable \cite{tarjan-scc}. Instead of measuring total cyclic mass, which is vulnerable to dilution by trivial SCCs, we isolate the single largest cyclic trap. Drawing inspiration from the phase transition and emergence of the ``giant component'' in classical random graph theory \cite{erdos-renyi}, our framework operates on the principle that the macroscopic acyclicity of a directed network is ultimately dictated by its single most massive feedback loop.

The locality of feedback, $M(G)$, is measured by penalizing the size of the largest SCC in the filtered graph:

$$M(G) = 1 - \frac{|V_{\text{max\_scc}}(G_{>2})| - 1}{|V_{>2}| - 1}$$

This formulation ensures that a DAG is mathematically recognized as only being as forward-flowing as its single largest cyclic trap, rendering the metric immune to dilution by peripheral acyclic nodes.

\textbf{Edge Case (DAGs).} 
If $G_{>2}$ contains no nontrivial strongly connected components, then 
$|V_{\max\_scc}(G_{>2})| = 1$, and we define $M(G) = 1$. 
This convention ensures that every directed acyclic graph attains the 
maximal locality-of-feedback score, consistent with the interpretation 
that a DAG contains no macroscopic cyclic traps.

\subsection{Pathway Complexity (Spectral Behavior)}

The final component acts as a dynamical tie-breaker, capturing the systemic chaos generated by pathway proliferation and undirected edges. Let $A$ be the adjacency matrix of the raw graph $G$. The spectral radius $\rho(A)$ provides a global indicator of recurrence. The pathway complexity, $S(G)$, is defined as:

$$S(G) = \frac{1}{1 + \rho(A)}$$

While the previous three dimensions analyze static topology, this component evaluates dynamic flow. A strictly one-way DAG possesses a nilpotent adjacency matrix and evaluates to exactly 1. Conversely, the introduction of symmetric 2-cycles yields real, non-zero eigenvalues, smoothly penalizing the score to reflect the increased dynamic complexity of back-and-forth diffusion.

The specific reciprocal normalization of $S(G)$ is chosen deliberately for its theoretical boundary conditions and asymptotic behavior. For a strict DAG, $\rho(A) = 0$, yielding $S(G) = 1.0$. For a simple directed cycle, $\rho(A) = 1$, yielding $S(G) = 0.5$. As the dynamical recurrence and pathway complexity of the network explode (e.g., in highly dense core-periphery structures), $\rho(A)$ grows continuously. The reciprocal form ensures that the penalty is bounded strictly within $[0, 1]$ while applying smoothly diminishing structural penalties for infinitely massive spectral radii.

\subsection{Composite DAG-ness Score}

The four components collectively isolate the volume of feedback, the hierarchical alignment of flow, the severity of macroscopic traps, and dynamical recurrence. Let $w = (w_A, w_F, w_M, w_S)$ be non-negative weights satisfying $w_A + w_F + w_M + w_S = 1$. The composite DAG-ness score is defined as:

$$D(G) = w_A A(G) + w_F F(G) + w_M M(G) + w_S S(G)$$

For general structural analysis, we adopt the uniform baseline weighting $w = (0.25, 0.25, 0.25, 0.25)$. The linear structure of $D(G)$ is designed intentionally so that future domain-specific applications may adjust the weights without altering the underlying orthogonal component definitions.

The formulation above is intentionally general. Different classes of real-world networks may exhibit characteristic 
patterns of cyclicity, hierarchy, or recurrence, and future work may therefore consider domain--specific or 
data--driven weight choices. The linear structure of $D(G)$ is designed to support such extensions 
without altering the underlying component definitions.
\section{Orthogonality and the Resolution of Component Redundancy}

The 4-dimensional framework introduced in Section 2 represents a significant structural departure from previous multi-component DAG-ness measures. In early formulations, such as the five-component framework previously explored by the author \cite{csikos2026}, DAG-ness was quantified by combining edge--level acyclicity, node-level cyclicity, macro-level cyclicity, hierarchical flow, and spectral recurrence into a single composite score. While that initial measure successfully demonstrated sensitivity to structural deviations and provided continuous values between 0 and 1, it suffered from mathematical entanglement, scalability limitations, and vulnerabilities to structural masking. The shift to a strictly orthogonal 4-dimensional framework resolves these theoretical limitations.

\subsection{Topological Redundancy and Collinearity}

A fundamental requirement for a multi-component metric is the orthogonality of its constituent parts. In the original framework, three of the five components were explicitly or implicitly tied to the presence of strongly connected components (SCCs) and cyclic edges \cite{csikos2026}. Specifically, the node-level cyclicity penalized the proportion of vertices in nontrivial SCCs, the macro--level cyclicity measured the fragmentation of the condensation DAG (which is derived directly from SCCs), and the hierarchy component penalized the proportion of edges residing on cycles \cite{csikos2026}.

Mathematically, these three components are highly collinear. If an edge forms part of a cycle, its endpoints must belong to a nontrivial SCC, and the presence of that SCC inherently reduces the number of nodes in the condensation graph \cite{bangjensen-gutin}. Consequently, a single cyclic intervention in the graph simultaneously inflates the penalty across the node, macro, and hierarchy dimensions, leading to a ``triple-counting'' of the exact same cyclic structures. 

The revised framework eliminates this redundancy by strictly separating the \textit{volume} of the feedback from its \textit{severity}. $A(G)$ quantifies the absolute minimum intervention required to eliminate all feedback loops using a heuristic minimum feedback arc set \cite{eades}, while $M(G)$ isolates the topological severity of the single largest cyclic trap. This guarantees that the components measure distinct graph properties without overlapping penalties.

\subsection{Scalability and the Trace--Exponential Walk Explosion}

To measure edge-level acyclicity, the original framework utilized a hybrid score that relied heavily on the trace--exponential surrogate $h(W) = Tr(e^{W \circ W}) - |V|$ \cite{csikos2026, zheng}. While this surrogate is differentiable and avoids the NP-hard computation of the exact minimum feedback arc set, it introduces a severe limitation when analyzing networks containing undirected (symmetric) relationships.

In algebraic graph theory, the trace of the $k$-th power of an adjacency matrix, $Tr(A^k)$, counts the exact number of closed walks of length $k$, not simple cycles \cite{meyer2000}. When an undirected edge is treated structurally as a pair of directed edges (a 2-cycle), it generates not only a closed walk of length 2, but phantom closed walks of lengths 4, 6, 8, and so on, as information bounces infinitely between the two vertices. Because the trace-exponential sums these traces, the penalty $h(W)$ explodes exponentially. This causes the composite score of simple undirected hierarchical trees to collapse toward 0, falsely categorizing them as highly cyclic non-DAGs.

By replacing the algebraic surrogate with the scalable Eades, Lin, and Smyth (ELS) heuristic \cite{eades} applied to the filtered graph $G_{>2}$, the new $A(G)$ component bypasses the walk explosion entirely. It operates in linear time $O(|V| + |E|)$ and accurately evaluates pure structural acyclicity.

In particular, undirected trees represented as symmetric digraphs now 
evaluate near $1$ under $D(G)$, rather than collapsing toward $0$ as in 
the trace--exponential formulation. This correction reflects the fact 
that symmetric diffusion alone does not constitute meaningful directed feedback.

\subsection{The Dilution Trap in Macro-Cyclicity}

In previous models, macro--acyclicity was defined by the ratio of trivial singleton SCCs to the total number of components in the condensation graph \cite{csikos2026}. This formulation is highly vulnerable to the ``Dilution Trap.'' 

Consider a directed network of 1000 nodes where 500 nodes are trapped in a single, catastrophic cyclic core, and the remaining 500 nodes are completely isolated or form simple forward--flowing chains on the periphery. The 500-node core collapses into 1 SCC, while the 500 peripheral nodes form 500 trivial singleton SCCs \cite{bangjensen-gutin}. Under the original formulation, the metric falsely reports a 99.8\% degree of macroscopic DAG-ness, completely masking the fact that 50\% of the system is locked in an inescapable feedback loop. 

The revised component $M(G)$ completely circumvents this by evaluating the Giant Component. By anchoring the penalty strictly to $|V_{\text{max\_scc}}|$, the network described above correctly incurs a 50\% penalty. This ensures the metric isolates the severity of the largest trap and remains immune to dilution by peripheral acyclic nodes.

\subsection{Disentangling Directedness from Acyclicity}

Finally, a true DAG is defined by both the absence of cycles and the presence of a progressive, hierarchical flow. The original framework defined hierarchy as the proportion of acyclic edges \cite{csikos2026}. However, this simply measures the proportion of acyclic edges, rendering it perfectly redundant with cycle--detection metrics. 

The new framework isolates directedness from acyclicity. While existing hierarchical metrics such as Trophic Coherence \cite{johnson2014} or Network Agony \cite{giechaskiel2020} provide powerful tools for evaluating systemic alignment, they are often mathematically entangled with the structural presence of cycles. By defining $F(G)$ strictly as the proportion of edges adhering to the acyclic topological gradient \cite{kahn1962}, the framework explicitly evaluates flow alignment independent of cycle detection. If a graph is perfectly acyclic but lacks a unified ``source-to-sink'' structure, $A(G)$ evaluates to 1 while $F(G)$ appropriately penalizes the lack of directional flow. This guarantees that $F(G)$ remains strictly orthogonal to the feedback components.

\subsection{Continuity of Empirical Behavior}

Although the internal definitions of the components have fundamentally changed to ensure mathematical orthogonality, the refined framework preserves the empirical ordering observed in the original study \cite{csikos2026}. In particular, the diagnostic graph families considered in that work (layered DAGs, layered DAGs with backedges, lollipop graphs, barbell graphs, and core-periphery graphs) appear in the exact same rank order under the refined $D(G)$ measure. 

Furthermore, the revised framework improves the theoretical behavior of the measure on complex mathematical systems, such as the Kaprekar graph and the Collatz graph. By separating the volume of acyclicity $A(G)$ from the dynamical recurrence $S(G)$, the new framework yields values that more accurately reflect the known structural and dynamical properties of these systems, successfully identifying them as massive, near-perfect directed trees pouring into isolated dynamical traps.

\subsection{Summary}

The refinements introduced in this paper preserve the conceptual intent of the original DAG-ness framework while replacing overlapping surrogates with definitions grounded in strictly orthogonal structural and dynamical properties. The resulting 4-dimensional measure is highly scalable, mathematically interpretable, immune to structural dilution, and remains fully compatible with the empirical narrative developed in \cite{csikos2026}. Finally, the linear composition of the framework remains entirely flexible, allowing for domain-specific or data-driven weight choices in future applied research without necessitating alterations to the underlying component definitions.
\section{Demonstration on Mathematical Graphs}

This section illustrates the behavior of the refined 4-dimensional DAG-ness framework on classical directed graphs arising from number-theoretic dynamical systems: the Kaprekar graph and the directed graph generated by the Collatz function. These examples highlight the ability of the orthogonal components to correctly diagnose networks that are topologically pure directed trees but contain isolated dynamical traps.

\subsection{Kaprekar's 6174 Graph}

The Kaprekar routine on four-digit numbers defines a directed map $K: \{0,\dots,9999\} \rightarrow \{0,\dots,9999\}$ by sorting the digits of $n$ in descending and ascending order, subtracting the latter from the former, and repeating. All nontrivial inputs eventually reach the fixed point 6174, which is the unique attractor of the system. The directed graph $G_{Kap}$ has 10,000 vertices and one outgoing edge per vertex. 

Under the orthogonal framework, $G_{Kap}$ evaluates as follows:
\begin{itemize}
    \item $A(G) = 1.000$: Removing exactly 1 edge breaks the single self-loop.
    \item $F(G) = 1.000$: The graph flows perfectly downhill into the attractor.
    \item $M(G) = 1.000$: The maximum cyclic trap (the self-loop) is only 1 node.
    \item $S(G) = 0.500$: The spectral radius of a 1-cycle is exactly 1.
\end{itemize}

\textbf{Composite Score:} $D(G_{Kap}) = 0.875$

The topological components ($A, F, M$) correctly recognize the structure as an essentially perfect macroscopic DAG. The penalty is isolated entirely to the pathway complexity dimension $S(G)$, reflecting the infinite dynamical recurrence occurring at the fixed-point attractor.

\subsection{The Collatz Graph}

The Collatz function $T: \mathbb{N} \rightarrow \mathbb{N}$ is defined by $T(n) = n/2$ if $n$ is even, and $3n+1$ if $n$ is odd. Restricting to the first 2228 positive integers yields a directed graph $G_{Col}$ with one outgoing edge per vertex. The only cyclic SCC is the well-known $1 \rightarrow 2 \rightarrow 4 \rightarrow 1$ cycle.

Under the orthogonal framework, $G_{Col}$ evaluates as follows:
\begin{itemize}
    \item $A(G) = 1.000$: Removing 1 edge breaks the 3-cycle.
    \item $F(G) = 1.000$: Only one edge must point uphill to close the terminal cycle.
    \item $M(G) = 0.999$: The largest trap is strictly 3 nodes ($1 - 2/2227$).
    \item $S(G) = 0.500$: The spectral radius of a simple 3-cycle is exactly 1.
\end{itemize}

\textbf{Composite Score:} $D(G_{Col}) = 0.875$

To understand why $S(G)$ evaluates to exactly $0.500$ despite the presence of only a single trivial cycle, we must consider the algebraic properties of the adjacency matrix $A$. The spectral radius $\rho(A)$ is determined by the maximum absolute eigenvalue of the system \cite{horn-johnson}. For a pure directed cycle of length $k$, the adjacency matrix acts as a permutation matrix, and its characteristic polynomial is $\lambda^k - 1 = 0$. 

For the 3-cycle in $G_{Col}$, the relevant submatrix is:
$$C_3 = \begin{pmatrix} 0 & 1 & 0 \\ 0 & 0 & 1 \\ 1 & 0 & 0 \end{pmatrix}$$
The roots of $\lambda^3 - 1 = 0$ are the cube roots of unity, all of which have an absolute magnitude of exactly $1$. 

Because the remaining 2,225 nodes form a strict directed tree feeding into this cycle, their corresponding submatrix is strictly nilpotent, yielding eigenvalues of exactly $0$. Consequently, the acyclic branches do not increase the maximum eigenvalue, leaving the global spectral radius at $\rho(A) = 1$. This yields:
$$S(G) = \frac{1}{1 + 1} = 0.500$$

This mathematical reality captures the crucial distinction between static topology and dynamic flow. Topologically, $G_{Col}$ is an overwhelming macroscopic DAG. Dynamically, however, any random walk on the graph will eventually fall into the terminal cycle and recur infinitely. The pathway complexity component isolates this dynamical entrapment perfectly, confirming that $G_{Col}$ and $G_{Kap}$ are fundamentally identical in their structural DAG-ness.

\subsection{The Modified Collatz Graph (Terminal Sink)}

To formally verify the boundary conditions of the framework, we construct a modified Collatz graph, $G_{Col}^*$, by defining the node $2$ as a terminal sink. By removing the $2 \rightarrow 1$ edge, the $1 \rightarrow 2 \rightarrow 4 \rightarrow 1$ loop is permanently severed. The graph becomes a pure collection of directed trees converging on a single endpoint.

Evaluating the modified graph $G_{Col}^*$:
\begin{itemize}
    \item $A(G) = 1.000$: The graph contains 0 cycles; the MFAS is 0.
    \item $F(G) = 1.000$: Every edge adheres strictly to a topological sort.
    \item $M(G) = 1.000$: There are no nontrivial strongly connected components.
    \item $S(G) = 1.000$: As a perfect directed tree, the adjacency matrix is nilpotent, yielding a spectral radius of 0.
\end{itemize}

\textbf{Composite Score:} $D(G_{Col}^*) = 1.000$

\medskip
This confirms that any mathematically pure directed acyclic graph evaluates to exactly $1.000$ under the continuous $D(G)$ framework.
To complement the qualitative analysis in Sections~5.1 and~5.2, we provide
a quantitative comparison of the DAG-ness components for the three
number-theoretic dynamical systems considered in this section: the
Kaprekar graph, the classical Collatz graph with its $1 \to 2 \to 4 \to 1$
cycle, and the modified Collatz variant. This table highlights how the
refined framework distinguishes between systems that are all nearly
acyclic at the macroscopic level but differ in their spectral and
macro-structural behavior.

\begin{table}[H]  % Requires \usepackage{float}
\centering
\begin{tabular}{lccccc}
\toprule
\textbf{Graph} & $A(G)$ & $F(G)$ & $M(G)$ & $S(G)$ & $D(G)$ \\
\midrule
Kaprekar (6174)           & 1.000 & 1.000 & 1.000 & 0.500 & 0.875 \\
Collatz (1--2--4 cycle)   & 1.000 & 1.000 & 0.999 & 0.500 & 0.875 \\
Modified Collatz          & 1.000 & 1.000 & 1.000 & 1.000 & 1.000 \\
\bottomrule
\end{tabular}

\caption{DAG-ness components for the Kaprekar graph, the classical
Collatz graph, and the modified Collatz variant. All values are normalized
to $[0,1]$.}
\label{tab:numbertheoretic-dagness}
\end{table}

\medskip
These results demonstrate that while all three systems exhibit nearly
perfect topological acyclicity, the refined components successfully capture
their distinct dynamical and macro-structural signatures, particularly in
the locality-of-feedback and spectral dimensions.
\section{Empirical Evaluation on Diagnostic Graphs}

To illustrate the diagnostic range and stability of the revised 4-dimensional framework, we evaluated it against the five synthetic graph families introduced in prior work. These diagnostic families isolate specific structural motifs---local cycles, global recurrence, dense strongly connected components (SCCs), hierarchical layering, and sparse back-edges---providing a controlled benchmark for continuous DAG-ness measures.

To ensure computational consistency and reproducibility, the experimental protocol from the original study was replicated. Each graph family was instantiated 20 times on $n=50$ vertices, with generation rules (such as edge density and cycle placement) subject to random variation. The four orthogonal components $A(G)$, $F(G)$, $M(G)$, and $S(G)$, along with the composite score $D(G)$, were computed for each instance. 

\subsection{Component-Level Results}

Table 1 reports the mean and standard deviation for the four components and the composite score across the 20 instances of each diagnostic family. The uniform baseline weighting $w = (0.25, 0.25, 0.25, 0.25)$ was applied.

\begin{table}[htbp]
\centering
\renewcommand{\arraystretch}{1.2}
\resizebox{\textwidth}{!}{%
\begin{tabular}{lccccc}
\hline
\textbf{Graph Family} & $\mathbf{A(G)}$ & $\mathbf{F(G)}$ & $\mathbf{M(G)}$ & $\mathbf{S(G)}$ & $\mathbf{D(G)}$ \\
\hline
Layered DAG + cycles & $0.983 \pm 0.01$ & $0.983 \pm 0.01$ & $0.957 \pm 0.01$ & $0.497 \pm 0.01$ & $\mathbf{0.855 \pm 0.01}$ \\
Layered DAG + backedges & $0.981 \pm 0.01$ & $0.971 \pm 0.01$ & $0.791 \pm 0.15$ & $0.409 \pm 0.05$ & $\mathbf{0.788 \pm 0.05}$ \\
Lollipop & $0.980 \pm 0.00$ & $0.980 \pm 0.00$ & $0.612 \pm 0.00$ & $0.500 \pm 0.00$ & $\mathbf{0.768 \pm 0.00}$ \\
Barbell SCCs + path & $0.837 \pm 0.03$ & $0.681 \pm 0.02$ & $0.720 \pm 0.01$ & $0.124 \pm 0.01$ & $\mathbf{0.591 \pm 0.01}$ \\
Core-periphery & $0.839 \pm 0.02$ & $0.664 \pm 0.01$ & $0.614 \pm 0.01$ & $0.080 \pm 0.00$ & $\mathbf{0.549 \pm 0.01}$ \\
\hline
\end{tabular}%
}
\caption{Refined 4-component DAG-ness values (mean $\pm$ std), averaged over 20 random instances of each diagnostic graph family.}
\label{tab:new_components}
\end{table}

The component-level results confirm the orthogonality and diagnostic clarity of the revised measure. The layered DAG with local cycles achieves near-perfect scores in acyclicity $A(G)$, flow $F(G)$, and macroscopic locality $M(G)$, losing points primarily in spectral complexity due to the presence of feedback. Conversely, the core-periphery and barbell graphs suffer massive penalties in $S(G)$ due to dense dynamical recurrence, while also showing severe structural degradation in alignment and acyclicity. 

The dataset beautifully highlights the separation of volume and locality. For example, the Lollipop graph requires almost zero edge removals to break its cycle ($A(G) = 0.980$), but the Giant Component penalty correctly flags the structural severity of that deterministic 20-node trap ($M(G) = 0.612$). In contrast, the Layered DAG with sparse back-edges exhibited a high standard deviation in its macroscopic penalty ($M(G) = 0.791 \pm 0.15$); this appropriately reflects the probabilistic nature of random back-edges, which sometimes create massive structural loops and other times create only minor localized traps.

\subsection{Comparative Analysis with Previous Framework}

To demonstrate the improvements of the orthogonal framework, Table 2 provides a direct comparison between the composite $D(G)$ scores of the original 5-component model and the revised 4-component model. 

\begin{table}[htbp]
\centering
\renewcommand{\arraystretch}{1.2}
\begin{tabular}{lcc}
\hline
\textbf{Graph Family} & \textbf{Original $D(G)$ } & \textbf{Revised $D(G)$} \\
\hline
Layered DAG + cycles & $0.932 \pm 0.129$ & $\mathbf{0.855 \pm 0.01}$ \\
Layered DAG + backedges & $0.495 \pm 0.159$ & $\mathbf{0.788 \pm 0.05}$ \\
Lollipop & $0.424 \pm 0.004$ & $\mathbf{0.768 \pm 0.00}$ \\
Barbell SCCs + path & $0.263 \pm 0.002$ & $\mathbf{0.591 \pm 0.01}$ \\
Core-periphery & $0.279 \pm 0.001$ & $\mathbf{0.549 \pm 0.01}$ \\
\hline
\end{tabular}
\caption{Comparison of empirical composite DAG-ness scores between the original 5-component framework and the revised 4-component framework.}
\label{tab:comparison}
\end{table}

The magnitude of the scores in the revised framework reflects a much more rigorous mathematical reality. In the original framework, the Lollipop graph was violently penalized ($D(G) \approx 0.424$) due to topological redundancy, where the single structural cycle triggered simultaneous, compounding penalties across node--level, macro--level, and hierarchical components. The revised framework correctly diagnoses it as predominantly DAG-like ($D(G) \approx 0.768$), isolating the penalty strictly to its single macroscopic trap and spectral properties. 

Furthermore, the revised framework exhibits significantly lower standard deviations across the randomized instances. By eliminating mathematically collinear components and brittle matrix-trace surrogates, the measure avoids the exponential score degradation caused by slight variations in local cycle density.
\section{Conclusion}

Quantifying the extent to which a directed network resembles a true DAG requires a measure that is both mathematically principled and empirically reliable. Previous attempts to define a continuous DAG-ness framework provided a critical conceptual foundation but suffered from topological redundancy. By relying on overlapping properties of strongly connected components and trace--exponential surrogates, prior frameworks were vulnerable to the Dilution Trap, exhibited high variance under probabilistic structural changes, and artificially compounded penalties on structurally distinct architectures \cite{csikos2026}.

The refined framework introduced in this paper resolves these theoretical limitations by establishing four strictly orthogonal dimensions of DAG-likeness: the volume of feedback $A(G)$, the alignment of flow $F(G)$, the macroscopic locality of feedback $M(G)$, and dynamical pathway complexity $S(G)$. By explicitly separating the absolute volume of cyclic edges from the topological severity of the Giant Component, the new measure avoids mathematical collinearity. Furthermore, evaluating spectral complexity strictly on the raw adjacency matrix while applying topological components to the filtered graph $G_{>2}$ allows the framework to successfully differentiate between simple symmetric diffusion and catastrophic cyclic traps.

Empirical evaluation on a suite of diagnostic graph families confirms that the refined 4-dimensional framework preserves the qualitative ordering of prior measures while drastically improving theoretical stability and interpretability. Additionally, deterministic evaluation on complex number-theoretic dynamical systems, such as the Kaprekar and Collatz graphs, demonstrates the framework's ability to accurately identify pure directed tree structures containing isolated dynamical attractors.

\end{document}